\documentclass[journal]{IEEEtran}

\ifCLASSINFOpdf
\else
   \usepackage[dvips]{graphicx}
\fi

\usepackage{cite}
\usepackage{amsmath,amssymb,amsfonts}
\usepackage{algorithmic}
\let\oldnl\nl
\newcommand{\nonl}{\renewcommand{\nl}{\let\nl\oldnl}}

\usepackage{graphicx}
\usepackage{textcomp}
\usepackage{enumitem}
\usepackage{subcaption}

\usepackage{xcolor}
\usepackage[colorlinks=true,  linkcolor=black, citecolor=blue, breaklinks=true, urlcolor=blue]{hyperref}

\usepackage{graphicx}
\usepackage{amssymb,amsmath,amsthm}
\usepackage[ruled,vlined,linesnumbered]{algorithm2e}
\usepackage{xcolor}
\usepackage{booktabs}
\usepackage[latin1]{inputenc}

\definecolor{darkgreen}{rgb}{0.2, 0.6, 0.05}
\usepackage{enumitem}

\usepackage{subcaption}

\newcounter{clemma}
\newtheorem{lemma}[clemma]{Lemma}

\makeatletter
\newcommand\notsotiny{\@setfontsize\notsotiny\@vipt\@viipt}
\makeatother

\usepackage{pgfplots}
\usepackage{tikz}
\usepgfplotslibrary{groupplots}
\pgfplotsset{compat=1.13}
\newlength\fwidth

\usepackage{algorithmic}

\usetikzlibrary{arrows,petri,topaths}

\newcommand{\E}{\mathbb{E}}                  
\newcommand{\bo}[1]{\mathbf{#1}}              
\newcommand{\bom}[1]{\boldsymbol{#1}}    

\newcommand{\pdim}{N}
\newcommand{\ndim}{L}

\newcommand{\hop}{\mathsf{H}}        
\newcommand{\beq}{\begin{equation}}
\newcommand{\eeq}{\end{equation}}
\newcommand{\bmat}{\begin{pmatrix}}
\newcommand{\emat}{\end{pmatrix}}

\renewcommand{\th}{\boldsymbol{\theta}}

\renewcommand{\S}{\hat{\M}} 
\newcommand{\C}{\mathbb{C}} 

\newcommand{\w}{\mathbf{w}}

\newcommand{\x}{\bo x}

\renewcommand{\a}{{\bo a}}
\newcommand{\Q}{\mathbf{Q}}

\newcommand{\M}{\bom \Sigma}

\DeclareMathOperator{\tr}{tr}

\DeclareMathOperator{\cov}{cov}

\begin{document}


\title{Greedy Capon Beamformer}
\author{Esa~Ollila,~\IEEEmembership{Senior Member,~IEEE}
\thanks{Esa Ollila is with the Department of Information and Communications Engineering, Aalto University, FI-00076 Aalto, Finland (e-mail: esa.ollila@aalto.fi).}}
\markboth{Greedy Capon Beamformer}
{Ollila \MakeLowercase{\textit{et al.}}}
\maketitle

\begin{abstract} 
We propose greedy  Capon beamformer (GCB) for direction finding of narrow-band sources present in the array's viewing field.  After defining the grid covering the  location search space, the algorithm greedily builds the interference-plus-noise covariance matrix by identifying a high-power source on the grid using Capon's principle of maximizing the signal to interference plus noise ratio while enforcing unit gain towards the signal of interest. An estimate of the power of the detected source is derived by exploiting the unit power constraint,  which subsequently  allows to update the  noise covariance matrix by simple rank-1 matrix addition composed of  outerproduct of the selected  steering matrix with itself  scaled by the signal power estimate.  Our numerical examples demonstrate  effectiveness of the proposed GCB in direction finding where it performs favourably  compared to the state-of-the-art algorithms under a broad variety of settings.  Furthermore, GCB estimates of direction-of-arrivals (DOAs) are very fast to compute. 
\end{abstract}

\IEEEpeerreviewmaketitle

\section{Introduction}
\label{sec:intro}

\IEEEPARstart{C}{ommon}  to most direction finding (DF) methods is the need to estimate the unknown $N \times N$ array covariance matrix $\M=\cov(\x) \succ 0$ of the array output $\x \in \mathbb{C}^{\pdim}$  of $N$-sensors. For example, the conventional  beamformer  \cite{van_trees:2002}  and the standard Capon 
beamformer (SCB) \cite{capon:1969} require an estimate of $\M$ to measure the power of the beamformer output as a function of the direction-of-arrival (DOA). 
In addition, many high-resolution  subspace-based DOA algorithms (such as MUSIC \cite{schmidt:1986} or R(oot)-MUSIC \cite{barabell1983improving}) 
compute  the  noise  or  signal  subspaces  from the eigenvectors of the array covariance matrix.  Array covariance matrix is conventionally estimated from the array snapshots via the sample covariance matrix (SCM). 
However, the SCM is poorly estimated when the snapshot size is small and adaptive beamformers (such as SCB) 
 are not computable when $L < N$.  Sparse methods for DOA estimation \cite{yang2018sparse} provide a remedy for these issues. 

We assume that $K$ narrowband sources  are present in the array's 
viewing field. Let $\a(\theta) \in \mathbb{C}^N$ denote the array manifold, where $\theta$ denotes the  generic location parameter in the location space $\Theta$.  For example,  $\a(\theta)=(1,e^{-\jmath \cdot 1 \cdot  \frac{2\pi d}{\lambda}\sin\theta}, \ldots, e^{-\jmath \cdot (N-1) \cdot  \frac{2\pi d}{\lambda}\sin\theta})^\top$ for Uniform Linear Array (ULA),  where $\lambda$ is the wavelength, $d$ is the element spacing between the sensors and $\theta \in  \Theta = [-\pi/2,\pi/2)$ is the DOA in radians. 
For an arbitrary steering vector in the array manifold we use notation $\a$, dropping the dependency on $\theta$. 
The output of a beamformer is defined by 
$
y = \w^\hop \x ,
$
where $\w$ is the beamformer weight vector that depends on $\theta$ through  $\a$. Ideally, the beamformer weight $\w=\w(\theta)$  is chosen such that the beamformer  will null the interferences and noise while allowing the signal of interest (SOI) at $\theta$ to pass undistorted.  A beamformer estimates the locations of the $K$ signals impinging on the array as $K$ peaks in   
the array output power distribution  \cite{krim_viberg:1996,van_trees:2002,elbir2023twenty}:
\[
P(\theta)=  \E[ | \w^{\hop} \x|^2] = \w^\hop \M \w
\]
using a fine grid $\{\theta_m\}_{m=1}^M$ covering $\Theta$, i.e., allowing  $\theta$ (and hence $\a$ in the design of the beamformer weight $\w$) vary through the location space $\Theta$. 


As sparse DOA methods leverage the covariance matrix's geometry, they  often outperform traditional approaches in low sample size (LSS) scenarios. Diagonal loading of the SCM is another common method to  combat LSS settings \cite{baggeroer1999passive,mestre2005finite,abramovich2007diagonally,yang2018high,richmond2020capon},  but these methods often involve increased computational complexity arising from determining the optimal loading factor.  In this letter, we propose a greedy Capon beamformer (GCB)  that greedily  selects the  next high-power source (not yet detected)  using Capon beamforming principle  and  subsequently updates the interference-plus-noise  covariance matrix (INCM).  Although GCB may seem heuristic, it is yet logical and follows generic greedy selection algorithmic framework. The approach is computationally light; it does not require computing the eigenvalue decomposition as needed by MUSIC or R-MUSIC nor inverting the SCM as in SCB (i.e., $L>N$ is not required).  There are similar sparsity and covariance based DOA estimation methods, e.g, \cite{wipf2007beamforming,stoica2010spice,yardibi2010source,abeida2012iterative,gerstoft2016multisnapshot,mecklenbrauker2024robust,ollila2024sparse,wang2021sparse}. or review in \cite{yang2018sparse}, which are primarily iterative algorithms constructed using some optimization principles.  Our simulation studies illustrate that the proposed GCB performs favourably against SCB and  state-of-the-art (SOTA) methods.

\section{Capon beamformer} \label{sec:SCB}

Let the  $N \times N$ array covariance matrix $\M$ be decomposed as  
\beq \label{eq:M}
\M = \gamma \a \a^\hop + \Q
\eeq 
where $\gamma >0$ denotes the power of the SOI, $\a \in \mathbb{C}^N $ is the array
steering vector of the SOI at  $\theta$, and $\Q$ is  an $N \times N$ positive definite ($\Q \succ 0$) 
INCM.  It is assumed that the noise is spatially white and $K-1$ uncorrelated directional interfering signals are present. Then the  matrix $\Q$  can be expressed in the form 
\[
\Q = \sum_{k=2}^K \gamma_k \a_k \a_k^\hop  + \sigma^2 \mathbf{I}
\]
where $\sigma^2 \mathbf{I}$ is an $N \times N$ noise covariance matrix (i.e., we assume spatially white noise), $\a_k = \a(\theta_k)\in \mathbb{C}^N $ is the steering vector of the directional interference signal  at $\theta_k$, 
and $\gamma_k>0$ is the associated signal power ($k=2,\ldots,K$).  Without any loss of generality (w.l.o.g.) 
we will assume that all steering vectors are normalized such that $\| \a_k \|^2 = N$ holds.  

Capon Beamformer (CB)  minimizes the array output power subject to the
constraint that the SOI is passed undistorted:
\beq \label{eq:crit1}
\min_{\w} \w^\hop \M \w \ \ \mbox{ subject to}  \ \w^\hop \a = 1, \tag{P1}
\eeq 
or equivalently, maximizes the signal to interference plus noise ratio (SINR) at its output while enforcing a unit gain
towards the SOI: 
\beq \label{eq:crit2}
\max_{\w} \frac{ \gamma |\w^\hop \a |^2}{\w^\hop \Q \w}  \ \ \mbox{ subject to}  \ \w^\hop \a = 1   \tag{P2}.
\eeq 
Due to the first formulation \eqref{eq:crit1},  CB is also often called minimum variance distortionless response (MVDR) beamformer. 
The solutions to power minimization problem \eqref{eq:crit1} and SINR maximization \eqref{eq:crit2} are equivalent  \cite{du2010fully} and the optimum beamformer weight for both problems is given by
\begin{align} 
\w_{\text{opt}} &=  \frac{\M^{-1}\a}{\a^\hop \M^{-1} \a}  = \frac{\Q^{-1}\a}{\a^\hop \Q^{-1} \a} .   \label{eq:wopt}
\end{align} 
The optimal power and SINR are then
\begin{align} 
P_{\text{opt}} &= \mathbb{E}[ |\w_{\textup{opt}}^\hop \x |^2] = \w_{\text{opt}}^\hop \M \w_{\text{opt}} = \frac{1}{\a^\hop \M^{-1} \a},  \label{eq:Popt} \\ 
\mathrm{SINR}_{\text{opt}}  &= \frac{ \gamma |\w^\hop_{\text{opt}} \a |^2}{\w^\hop_{\text{opt}} \Q \w_{\text{opt}}}  = \gamma \a^\hop \Q^{-1} \a.
\end{align} 

So far we have discussed the ideal situation, so assuming that $\theta$ of SOI (and hence the steering vector $\a$) as well as the covariance matrix $\M$ are known exactly. 
Since $\M$  is
unavailable in practise, it is commonly estimated by the SCM $\S  =  \frac{1}{L} \sum_{l=1}^{\ndim} \x_{l} \x_l^\hop 
\in \C^{\pdim \times \pdim}$, where 
$\x_l  \in \C^{\pdim}$ represents the $l^{\text{th}}$ snapshot and $L$ denotes the number of snapshots (i.e., sample size).
Consider a grid $\{\theta\}_{m=1}^M$,   $\theta_1 < \ldots  < \theta_M$, of location parameters covering $\Theta$,  and  let  $\{\a_m\}_{m=1}^M$ denote the  corresponding steering vectors. Write 
 $\th=(\theta_1,\ldots,\theta_M)^\top$ for  the $M$-vector of DOA parameters on the grid.  Then SCB computes an estimate of the spatial power for SOI at $\theta_i$  using 
$\hat{P}_{\text{SCB},i} = ( \a^\hop_i  \hat{\M}^{-1} \a_i )^{-1}$,   $i = 1, \ldots, M$
which is an  empirical (sample based) estimate of \eqref{eq:Popt}.  
Typically, it is assumed that the grid resolution is fine enough
so that the true location parameters of the sources lie  in close proximity to the grid. 
SCB then identifies the  indices of the  largest $K$ peaks in the spatial spectrum:
\[
\mathcal K = \mathsf{findpeaks}_K(\hat{P}_{\text{SCB},1} , \ldots, \hat{P}_{\text{SCB},M}) 
\]
where $\mathcal K \subset \{1,\ldots,M \}$ with  $|\mathcal K | = K$ and outputs the vector of $K$  DOA estimates as $\hat{\th} = \th_{\mathcal K}$, 
where $\th_{\mathcal K}$  denotes a $K$-vector consisting of components  of an $M$-vector $\th$ corresponding to indices in 
the set  $\mathcal K$. 

\section{The Greedy Capon beamformer} 

 First we prove a result that relates the power $\gamma$ of the SOI to the optimum power and optimum SINR  of CB.

\begin{lemma}  \label{lem:CBF} The power $\gamma$ of the SOI at steering vector $\a$ is given by 
\beq \label{eq:lem:gamma}
\gamma =  \w_{\textup{opt}}^\hop (\M-\Q)  \w_{\textup{opt}} = P_{\textup{opt}} - (\a^\hop \Q^{-1} \a)^{-1}
\eeq 
where $P_{\textup{opt}}$ is the optimum beamformer power \eqref{eq:Popt}.
\end{lemma}

\begin{proof} Due to unit gain constraint one has that  $\w_{\text{opt}}^\hop \a = 1$  and since $\gamma \a \a^\hop = \M-\Q$  from \eqref{eq:M},  we may write 
\begin{align} 
\gamma &= \gamma |  \w_{\text{opt}}^\hop \a |^2 =    \w_{\text{opt}}^\hop( \gamma \a  \a^{\hop} )  \w_{\text{opt}} = \w_{\text{opt}}^\hop (\M - \Q)  \w_{\text{opt}}   \notag   \\
&= \w_{\text{opt}}^\hop \M   \w_{\text{opt}}   -  \w_{\text{opt}}^\hop \Q   \w_{\text{opt}}  =  P_{\textup{opt}} -  (\a^\hop \Q^{-1}   \a)^{-1}
\end{align} 
where the last identity follows by recalling \eqref{eq:Popt} and  noting that 
$\w_{\text{opt}}^\hop \Q   \w_{\text{opt}} = (\a^\hop \Q^{-1}   \a )^{-1} = (\mathrm{SINR}_{\text{opt}}/\gamma)^{-1}  $, 
when  the latter form  of $\w_{\text{opt}}$ in \eqref{eq:wopt} is invoked. 
\end{proof} 

Lemma~1 also illustrates  $P_{\text{opt}}$ is essentially equal to the power $\gamma$ of the SOI in ideal conditions. Namely,  if the INCM $\Q$ contains only the white noise term,  i.e.,  $\Q = \sigma^2 \mathbf{I}$,  then $\gamma =   P_{\textup{opt}}  -  \frac{\sigma^2}{N}$ where we used that $\a^\hop \a = N$. If $\gamma \gg \sigma^2$ and $N$ is reasonably large, then we can notice that $\gamma \approx P_{\textup{opt}}$ which is an expected result. 

\autoref{lem:CBF} will now be used in our greedy Capon beamformer (GCB) algorithm to estimate the power of the detected signal. 
GCB greedily builds the INCM $\Q$ by identifying a high-power source on the grid using Capon's principle of maximizing the SINR while enforcing unit gain towards the SOI.  The {\bf GCB algorithm} proceeds as follows: 

{\bf Input:}  The SCM $\hat \M$, the number of sources $K$, the grid of DOAs $\th=(\theta,\ldots,\theta_M)^\top$ and the steering vectors $\a_1, \ldots , \a_M$.

{\bf Initialize}: Set $\hat{\Q} = [\tr(\S)/\pdim ] \cdot \mathbf{I}$ and $\mathcal M = \emptyset$. Initially, the set $\mathcal{M}$, which gathers the identified sources, is empty, and the INCM $\hat{\mathbf{Q}}$ is merely an estimate of the white noise matrix $\sigma^2 \mathbf{I}$. 

{\bf Main iteration}:   for $k=1,\ldots,K$, iterate the steps below
\begin{enumerate} 
\item Compute the beamformer output power estimates
\[
\hat P_i =  \w_i^\hop \S \w_i =  \dfrac{\a_i^\hop \hat{\Q}^{-1}   \S  \hat{\Q}^{-1} \a_i}{(\a_i^\hop \hat{\Q}^{-1} \a_i)^2},  \ i = 1,\ldots,M, 
\]
where $\w_i = \hat{\Q}^{-1} \a_i  / (\a_i^\hop \hat{\Q}^{-1} \a_i)$ signifies the beamformer weight for location $\theta_i$ on the grid. 

\item Identify the  indices of  the $k$ largest peaks  in the spatial spectrum:
\[
\mathcal K = \mathsf{findpeaks}_k( \hat P_1, \ldots, \hat P_M) , 
\]
i.e.,  $\mathcal K \subset \{1,\ldots,M \}$ with cardinality $|\mathcal K | = k$. 

\item If $k<K$, then execute steps 4) - 7).

\item Choose the index that has the least coherence with steering vectors that have been  picked up so far:
\[
i_k = \arg \min_{i \in  \mathcal K}  (\max_{j \in \mathcal M} | \a_i^\hop \a_j | ).
\]
Thus, $i_k$ is the index from the set $\mathcal{K}$ corresponding to the steering vector $\mathbf{a}_{i_k}$, which is the least coherent with the previously selected steering vectors $\mathbf{a}_j$ for $j \in \mathcal{M}$. 

\item Update the  set $\mathcal M   \gets \mathcal M  \cup \{ i_k \}$ of chosen indices. 

\item Estimate the signal power of the chosen source as $\hat{\gamma}_k = \hat P_{i_k} - (\a^\hop_{i_k} \hat{\Q}^{-1} \a_{i_k})^{-1}$, i.e., using equation \eqref{eq:lem:gamma} of Lemma~\ref{lem:CBF}. 

\item Update the INCM as  $\hat{\Q}  \gets \hat{\Q} + \hat{\gamma}_{k} \a_{i_k} \a_{i_k}^\hop$.  Such rank-1 update allows to compute 
the inverse covariance matrix $ \hat{\Q}^{-1}$ efficiently using the Sherman-Morrison \cite{sherman1950adjustment} formula, 
yielding
\beq 
\hat{\Q}^{-1} \gets  \hat{\Q}^{-1}   -  \dfrac{ \hat \gamma_k  \hat{\Q}^{-1}   \a_{i_k}  \a_{i_k}^{\hop} \hat{\Q}^{-1}  }{ 1+ \hat \gamma_{k}   \a_{i_k}^\hop \hat{\Q}^{-1} \a_{i_k}}. 
\eeq
\end{enumerate} 

{\bf Output}: After $K$ iterations,  the DOA estimates are given as $\hat{\th} = \th_{\mathcal K}$ similarly as  in the SCB algorithm outlined in \autoref{sec:SCB}. 

MATLAB and python functions, and simulation implementations, are available at \url{https://github.com/esollila/GCB}.

The dominant complexity of GCB algorithm is in the matrix-vector 
multiplications of step 1, whose complexity is $O(N^2)$. This is repeated for all $M$ steering vectors, so total  complexity is $O(MN^2)$. 
The complexity of step 4 is $O(k K N)$, where $k < K$ 
while the complexity of the  peak finding algorithm
in step 2 is $O(M)$. The complexity of rank-one update in step~7 is only $O(N^2)$ instead of $O(N^3)$ due to inverting an $N \times N$ matrix. Hence the overall complexity of GCB algorithm is $O(KMN^2)$.  The  complexity of SCB and MUSIC  is $O(MN^2 + N^3)$ which is typically slightly better than that of GCB since their complexity reduces to $O(MN^2)$ when   $M \gg N$.  However, both  the SCB and MUSIC need to assume $L>N$ which is not needed by GCB.  The computational complexity of IAA(-APES) \cite[Table~2]{yardibi2010source}, a widely used sparse DOA estimation method, is $O(N_{iter}( 2 MN^2 + N^3))$, where $N_{iter}$ denotes the number of iterations required for convergence. This complexity is generally higher than that of GCB, but similar magnitude when $N_{iter}< N \ll M$.  
 
Average running times of different methods  on a MacBook Pro M3 laptop in the simulation set-up of Section~IV  ($K=4$, $N=20$, $M=1801$, $L=125$) are compared in \autoref{tab:running_times}. 
Running times are averages over 100 MC simulations and over different SNR levels of Figure~1. 
We compared with the following SOTA methods: 
IAA, SPICE+ \cite[eq. (44)-(46)]{stoica2010spice}, SBL1 \cite[Table~1]{gerstoft2016multisnapshot}  and CGDP(-SBL)1 \cite[Table~I]{wang2021sparse}. For the last two methods, we used the efficient implementations available at the authors github pages \!\cite{GDP-SBL,SBL} while for IAA and SPICE+, we used our own implementation. 
As can be noted, CGBD1 and SBL1  are  3 orders (while SPICE+  2 orders) of  magnitude  slower than GCB, SCB, MUSIC,  R-MUSIC, and IAA, i.e., taking seconds instead of milliseconds to compute. 
\vspace{-3pt}

\begin{table}[!h]
 \centerline{
%
%
\definecolor{mycolor1}{rgb}{0.00000,0.44700,0.74100}%
\begin{tikzpicture}

\begin{axis}[%
width= 0.9\columnwidth, 
height= 2.6cm, 
 y tick label style={font=\scriptsize} , 
 x tick label style={font=\notsotiny} , 
scale only axis,
bar shift auto,
xmin=-0.2,
xmax=9.0,
xtick={1, 2, 3, 4, 5, 6, 7, 8},
ymin=0,
ymax=9.9,
ylabel style={font=\footnotesize\color{white!15!black}},
ylabel={Running time [s]},
xticklabels = {GCB, SCB,MUSIC, R-MUSIC, IAA, SPICE+, SBL1, CGDP1},
axis background/.style={fill=white},
axis x line*=bottom,
axis y line*=left,
legend style={legend cell align=left, align=left, draw=white!15!black}
]
\addplot[ybar, bar width=0.8, fill=mycolor1, draw=black, area legend] table[row sep=crcr] {%
1	0.002962515735\\
2	0.00122037690083333\\
3	0.00153122958333333\\
4	0.0014574058825\\
5	0.00899482269166667\\
6	0.186013558778333\\
7	3.71684739413333\\
8	8.90318417679833\\
};
\addplot[forget plot, color=white!15!black] table[row sep=crcr] {%
-0.2	0\\
9.2	0\\
};

\node[above, align=center]
at (axis cs:1,0.003) {\tiny $3.0e^{-3}$};
\node[above, align=center]
at (axis cs:2,0.001) {\tiny $1.2e^{-3}$};
\node[above, align=center]
at (axis cs:3,0.002) {\tiny $1.5e^{-3}$};
\node[above, align=center]
at (axis cs:4,0.001) {\tiny $1.5e^{-3}$};
\node[above, align=center]
at (axis cs:5,0.009) {\tiny $9.0e^{-3}$};
\node[above, align=center]
at (axis cs:6,0.186) {\tiny $1.9e^{-1}$};
\node[above, align=center]
at (axis cs:7,3.717) {\tiny $3.7$};
\node[above, align=center]
at (axis cs:8,8.903) {\tiny $8.9$};
\end{axis}
\end{tikzpicture}%
}
 \caption{Average running times of methods on a laptop.} \label{tab:running_times}
\end{table}
 

Another advantage of GCB over SCB is its robustness to coherent sources. 
While coherent sources make the SCM poorly conditioned, the GCB estimate of the INCM $\mathbf{Q}$ remains well-conditioned, eliminating the need for mitigation techniques like spatial smoothing.

 \section{Simulation study}   \label{sec:conclusion}

 The array is ULA with half a wavelength inter-element spacing and number of sensors is $N=20$. 
The grid size is $M=1801$, and spacing is uniform, thus providing angular resolution  $\Delta \theta=0.1^o$. The number of Monte-Carlo (MC) trials is $15000$ in order to obtain accurate finite sample performance and smooth curves.  The SNR  of $k^{\text{th}}$ source is defined as $\mathrm{SNR}_k = \gamma_k/\sigma^2$, where 
 $\gamma_k $ is the signal power of the  $k^{\text{th}}$ source and $\sigma^2>0$ is the noise variance. The array SNR is defined as average 
$\mbox{SNR (dB)} = \frac{10}{K}  \sum_{k=1}^K  \log_{10} \mbox{SNR}_k$.   The $K$ sources are following complex circular Gaussian distribution, $s_k \sim \mathcal C \mathcal N(0, \gamma_k)$,  $k=1, \ldots, K$,  unless otherwise noted.

\begin{figure}[!t]
\centerline{
%
%
\definecolor{mycolor1}{rgb}{1.00000,0.00000,1.00000}%
\definecolor{mycolor2}{rgb}{0.34200,0.87200,0.50300}%
\definecolor{mycolor3}{rgb}{0.92900,0.69400,0.12500}%
\begin{tikzpicture}

\begin{axis}[%
width= 0.44\columnwidth, 
height= 4.3cm, 
 y tick label style={font=\scriptsize} , 
 x tick label style={font=\scriptsize} , 
 scale only axis,
xmin=-14,
xmax=-8.5,
xlabel style={font=\scriptsize\color{white!15!black}},
xlabel={SNR (dB)},
ymode=log,
ymin=0.326413432578675,
ymax=33.863003001309,
yminorticks=true,
ylabel style={font=\color{white!15!black}},
axis background/.style={fill=white},
title style={font=\bfseries},
xmajorgrids,
ymajorgrids,
yminorgrids,
legend style={legend cell align=left, align=left, draw=white!15!black,font=\scriptsize,at={(1.0,1.0)},opacity=0.92}
]
\addplot [color=blue, line width=0.6pt, mark size=1.9pt, mark=o, mark options={solid, blue}]
  table[row sep=crcr]{%
-8.5	0.344162171076367\\
-9	0.369009846300428\\
-9.5	0.396364983317143\\
-10	0.427045665005511\\
-10.5	0.461473581764618\\
-11	0.500158774790565\\
-11.5	0.837541003971348\\
-12	1.39512231243955\\
-12.5	2.59566803219005\\
-13	4.51980639703369\\
-13.5	7.31238747149883\\
-14	10.0194951103669\\
};
\addlegendentry{GCB}

\addplot [color=mycolor1, line width=0.6pt, mark size=2.6pt, mark=+, mark options={solid, mycolor1}]
  table[row sep=crcr]{%
-8.5	0.361010249161986\\
-9	0.389193525125996\\
-9.5	0.420756144736274\\
-10	0.45540677787373\\
-10.5	0.643498562546957\\
-11	0.868824953601127\\
-11.5	1.42383697568694\\
-12	2.7060441730812\\
-12.5	4.63283429446813\\
-13	6.48838665103531\\
-13.5	9.36649041352556\\
-14	12.4496619713147\\
};
\addlegendentry{SCB}

\addplot [color=mycolor2, line width=0.6pt, mark size=2.6pt, mark=diamond, mark options={solid, mycolor2}]
  table[row sep=crcr]{%
-8.5	0.350171005843335\\
-9	0.380510709441928\\
-9.5	0.550489842473651\\
-10	0.89357849869686\\
-10.5	1.35000503702764\\
-11	2.68636348992463\\
-11.5	4.93774262863777\\
-12	6.88125051619737\\
-12.5	9.39651780182423\\
-13	12.5685403130196\\
-13.5	15.6392418145297\\
-14	18.6218942215876\\
};
\addlegendentry{MUSIC}

\addplot [color=mycolor3, dashed, line width=0.6pt, mark size=2.3pt, mark=triangle, mark options={solid, rotate=180, mycolor3}]
  table[row sep=crcr]{%
-8.5	0.344328612519636\\
-9	0.374972791061185\\
-9.5	0.544668284603903\\
-10	0.902618512943028\\
-10.5	1.90158502621501\\
-11	3.6042188485013\\
-11.5	6.02269230744359\\
-12	8.82474484569228\\
-12.5	12.1344705581284\\
-13	15.4661209020972\\
-13.5	18.7691133244484\\
-14	22.199598644369\\
};
\addlegendentry{R-MUSIC}

\addplot [color=red, dashed, line width=0.6pt, mark size=1.8pt, mark=square, mark options={solid, red}]
  table[row sep=crcr]{%
-8.5	0.349651826822053\\
-9	0.37569313719222\\
-9.5	0.403851457840626\\
-10	0.436058405873952\\
-10.5	1.67405559445717\\
-11	3.44953430673379\\
-11.5	6.59052710082179\\
-12	10.5142565436966\\
-12.5	15.1757593417924\\
-13	21.2603635340509\\
-13.5	27.6644924117541\\
-14	33.863003001309\\
};
\addlegendentry{IAA}

\addplot [color=black, line width=0.8pt]
  table[row sep=crcr]{%
-8.5	0.326413432578675\\
-9	0.352447368950927\\
-9.5	0.381121974064003\\
-10	0.412760088446054\\
-10.5	0.447726141825284\\
-11	0.486431236697659\\
-11.5	0.529338812425413\\
-12	0.576970961121026\\
-12.5	0.629915476887484\\
-13	0.68883373140509\\
-13.5	0.754469481405001\\
-14	0.827658727316119\\
};

\end{axis}
\end{tikzpicture}%
\hspace{-6pt}
%
%
\definecolor{mycolor1}{rgb}{1.00000,0.00000,1.00000}%
\definecolor{mycolor2}{rgb}{0.34200,0.87200,0.50300}%
\definecolor{mycolor3}{rgb}{0.92900,0.69400,0.12500}%
\begin{tikzpicture}

\begin{axis}[%
width= 0.44\columnwidth, 
height= 4.3cm, 
 y tick label style={font=\scriptsize} , 
 x tick label style={font=\scriptsize} , 
scale only axis,
xmin=-8.75,
xmax=-2.75,
xlabel style={font=\scriptsize\color{white!15!black}},
xlabel={SNR (dB)},
ymode=log,
ymin=0.329082705762752,
ymax=28.0996181255191,
yminorticks=true,
ylabel style={font=\small\color{white!15!black}},
axis background/.style={fill=white},
title style={font=\bfseries},
xmajorgrids,
ymajorgrids,
yminorgrids,
legend style={legend columns = 6, legend cell align=left, align=left, draw=white!15!black,font=\small,at={(0.98,1.21)}}
]
\addplot [color=blue, line width=0.6pt, mark size=1.9pt, mark=o, mark options={solid, blue}]
  table[row sep=crcr]{%
-2.75	0.364008058518854\\
-3.25	0.384885610712237\\
-3.75	0.407439811505943\\
-4.25	0.432016512030114\\
-4.75	0.459518153431758\\
-5.25	0.504447420451328\\
-5.75	0.720819071149851\\
-6.25	0.942067938102132\\
-6.75	1.53198128796232\\
-7.25	2.10313905072078\\
-7.75	3.06781570502533\\
-8.25	4.67824531208016\\
-8.75	6.80763245188809\\
};

\addplot [color=mycolor1, line width=0.6pt, mark size=2.2pt, mark=+, mark options={solid, mycolor1}]
  table[row sep=crcr]{%
-2.75	3.50508953761052\\
-3.25	4.64660044907385\\
-3.75	5.91153752137856\\
-4.25	7.21260918114934\\
-4.75	9.00809816405955\\
-5.25	10.7225997096475\\
-5.75	12.7849451726109\\
-6.25	15.2768718089361\\
-6.75	17.6101916930698\\
-7.25	20.0710751746553\\
-7.75	22.5566889089098\\
-8.25	25.2882400468413\\
-8.75	28.0996181255191\\
};

\addplot [color=mycolor2, line width=0.6pt, mark size=2.4pt, mark=diamond, mark options={solid, mycolor2}]
  table[row sep=crcr]{%
-2.75	0.372139937837009\\
-3.25	0.398155246103828\\
-3.75	0.462633476235057\\
-4.25	0.972223842538334\\
-4.75	1.09481426126383\\
-5.25	1.73298193874028\\
-5.75	2.57277214951759\\
-6.25	3.89567727958396\\
-6.75	5.0098551143388\\
-7.25	6.79220195027601\\
-7.75	8.68819672122279\\
-8.25	10.5697951856536\\
-8.75	12.7805518712352\\
};

\addplot [color=mycolor3, dashed, line width=0.6pt, mark size=2.4pt, mark=triangle, mark options={solid, rotate=180, mycolor3}]
  table[row sep=crcr]{%
-2.75	0.367300730513215\\
-3.25	0.393590155429302\\
-3.75	0.812342792934154\\
-4.25	1.22215871848487\\
-4.75	1.73398575481083\\
-5.25	2.52918235514634\\
-5.75	3.6507763814799\\
-6.25	5.28461837870845\\
-6.75	6.91986782702758\\
-7.25	9.31804977245866\\
-7.75	11.5785659797001\\
-8.25	13.8385784709022\\
-8.75	16.8593950623336\\
};

\addplot [color=red, dashed, line width=0.6pt, mark size=1.8pt, mark=square, mark options={solid, red}]
  table[row sep=crcr]{%
-2.75	0.366832023320392\\
-3.25	0.389675420489069\\
-3.75	0.415272280156846\\
-4.25	0.441981900081892\\
-4.75	0.674458894225587\\
-5.25	0.708186604410637\\
-5.75	1.51054952473154\\
-6.25	1.9159834724409\\
-6.75	2.77553922688908\\
-7.25	4.48916875750214\\
-7.75	6.18720439832616\\
-8.25	8.84061522746012\\
-8.75	11.6051059624632\\
};

\addplot [color=black, line width=0.8pt]
  table[row sep=crcr]{%
-2.75	0.329082705762752\\
-3.25	0.350964326140912\\
-3.75	0.374570075178436\\
-4.25	0.400076240589812\\
-4.75	0.427681673426011\\
-5.25	0.45761091462675\\
-5.75	0.490117726801828\\
-6.25	0.525489074035939\\
-6.75	0.564049595870257\\
-7.25	0.606166625474526\\
-7.75	0.652255806601878\\
-8.25	0.70278736947276\\
-8.75	0.758293132549827\\
};

\end{axis}
\end{tikzpicture}%
}
\caption{RMSE  of $\hat{\boldsymbol{\theta}}$ (deg.) versus SNR when $L=125$ (left panel) and  $L=25$ (right panel); $K=4$, $N=20$,  $M=1801$.}    \label{fig:4source}
\centerline{
%
%
\definecolor{mycolor1}{rgb}{1.00000,0.00000,1.00000}%
\definecolor{mycolor2}{rgb}{0.34200,0.87200,0.50300}%
\definecolor{mycolor3}{rgb}{0.92900,0.69400,0.12500}%
\begin{tikzpicture}

\begin{axis}[%
width= 0.43\columnwidth, 
height= 4.3cm, 
 y tick label style={font=\scriptsize} , 
 x tick label style={font=\scriptsize} , 
 at={(0\columnwidth,0\columnwidth)},
scale only axis,
xmin=-14,
xmax=-8.5,
xlabel style={font=\scriptsize\color{white!15!black}},
xlabel={SNR (dB)},
ymin=0,
ymax=25.6133333333333,
ylabel style={font=\color{white!15!black}},
title style={font=\small},
title={a) percentage of misdetection},
axis background/.style={fill=white},
legend style={legend cell align=left, align=left, draw=white!15!black}
]
\addplot [color=blue, mark=o, line width=0.6pt, mark options={solid, blue}]
  table[row sep=crcr]{%
-8.5	0\\
-9	0\\
-9.5	0\\
-10	0\\
-10.5	0\\
-11	0\\
-11.5	0.02\\
-12	0.106666666666667\\
-12.5	0.393333333333333\\
-13	1.12\\
-13.5	2.62666666666667\\
-14	4.94\\
};

\addplot [color=mycolor1, mark=+, line width=0.6pt, mark size=2.5pt, mark options={solid, mycolor1}]
  table[row sep=crcr]{%
-8.5	0\\
-9	0\\
-9.5	0\\
-10	0\\
-10.5	0.00666666666666667\\
-11	0.02\\
-11.5	0.1\\
-12	0.346666666666667\\
-12.5	0.986666666666667\\
-13	2.04\\
-13.5	4.16\\
-14	7.26\\
};

\addplot [color=mycolor2, mark=diamond,   line width=0.6pt, mark size=2.3pt, mark options={solid, mycolor2}]
  table[row sep=crcr]{%
-8.5	0\\
-9	0\\
-9.5	0.00666666666666667\\
-10	0.02\\
-10.5	0.1\\
-11	0.386666666666667\\
-11.5	1.10666666666667\\
-12	2.35333333333333\\
-12.5	4.46666666666667\\
-13	7.61333333333333\\
-13.5	11.7066666666667\\
-14	16.7533333333333\\
};

\addplot [color=mycolor3, dashed, mark=triangle,  line width=0.6pt, mark size=2.3pt, mark options={solid, rotate=180, mycolor3}]
  table[row sep=crcr]{%
-8.5	0\\
-9	0\\
-9.5	0.00666666666666667\\
-10	0.0333333333333333\\
-10.5	0.2\\
-11	0.746666666666667\\
-11.5	2.02666666666667\\
-12	4.31333333333333\\
-12.5	7.82666666666667\\
-13	12.9133333333333\\
-13.5	18.6866666666667\\
-14	25.6133333333333\\
};

\addplot [color=red, dashed, mark=square,  line width=0.6pt,mark size=1.8pt, mark options={solid, red}]
  table[row sep=crcr]{%
-8.5	0\\
-9	0\\
-9.5	0\\
-10	0\\
-10.5	0.0466666666666667\\
-11	0.213333333333333\\
-11.5	0.78\\
-12	2.02\\
-12.5	4.23333333333333\\
-13	8.50666666666667\\
-13.5	14.6\\
-14	22.0733333333333\\
};

\end{axis}

\begin{axis}[%
width= 0.26\columnwidth, 
height= 2.5cm, 
 y tick label style={font=\tiny} , 
 x tick label style={font=\tiny} , 
scale only axis,
at={(0.15\columnwidth,0.19\columnwidth)},
xtick = {-10,-11},
scale only axis,
xmin=-11.75,
xmax=-9.75,
ymin=0,
ymax=2.98,
axis background/.style={fill=white},
]
\addplot [color=blue, mark=o,  line width=0.6pt,mark options={solid, blue}]
  table[row sep=crcr]{%
-9.75	0\\
-10.25	0\\
-10.75	0\\
-11.25	0.0133333333333333\\
-11.75	0.0666666666666667\\
};

\addplot [color=mycolor1, mark=+, line width=0.6pt,mark size=2.5pt, mark options={solid, mycolor1}]
  table[row sep=crcr]{%
-9.75	0\\
-10.25	0.00666666666666667\\
-10.75	0.0133333333333333\\
-11.25	0.0533333333333333\\
-11.75	0.226666666666667\\
};

\addplot [color=mycolor2, mark=diamond,   line width=0.6pt,mark size=2.3pt, mark options={solid, mycolor2}]
  table[row sep=crcr]{%
-9.75	0.00666666666666667\\
-10.25	0.06\\
-10.75	0.2\\
-11.25	0.666666666666667\\
-11.75	1.58666666666667\\
};

\addplot [color=mycolor3, dashed, mark=triangle,   line width=0.6pt,mark size=2.3pt, mark options={solid, rotate=180, mycolor3}]
  table[row sep=crcr]{%
-9.75	0.0133333333333333\\
-10.25	0.1\\
-10.75	0.413333333333333\\
-11.25	1.29333333333333\\
-11.75	2.98\\
};

\addplot [color=red, dashed, mark=square,  mark size=1.8pt,  line width=0.6pt, mark options={solid, red}]
  table[row sep=crcr]{%
-9.75	0\\
-10.25	0.00666666666666667\\
-10.75	0.0933333333333333\\
-11.25	0.46\\
-11.75	1.28\\
};

\end{axis}
\end{tikzpicture}%
  \hspace{-9pt} 
%
%
\definecolor{mycolor1}{rgb}{1.00000,0.00000,1.00000}%
\definecolor{mycolor2}{rgb}{0.34200,0.87200,0.50300}%
\definecolor{mycolor3}{rgb}{0.92900,0.69400,0.12500}%
\begin{tikzpicture}

\begin{axis}[%
width= 0.43\columnwidth, 
height= 4.3cm, 
 y tick label style={font=\scriptsize} , 
 x tick label style={font=\scriptsize} , 
scale only axis,
xmin=-14,
xmax=-8.5,
xlabel style={font=\scriptsize\color{white!15!black}},
xlabel={SNR (dB)},
ymode=log,
ymin=0.123503866390242,
ymax=0.302279120460986,
yminorticks=true,
ylabel style={font=\color{white!15!black}},
axis background/.style={fill=white},
title style={font=\small},
title={b) RMSE of $\hat \theta_1$},
xmajorgrids,
ymajorgrids,
yminorgrids,
legend style={legend cell align=left, align=left, draw=white!15!black,font=\scriptsize,at={(1.0,1.0)}}
]
\addplot [color=blue, mark=o, mark options={solid, blue}] 
  table[row sep=crcr]{%
-8.5	0.130033329060924\\
-9	0.137983090751488\\
-9.5	0.146733318188702\\
-10	0.156356430418878\\
-10.5	0.166861219780592\\
-11	0.177685114739531\\
-11.5	0.190537835262887\\
-12	0.204921123036809\\
-12.5	0.220902693510062\\
-13	0.23817360615036\\
-13.5	0.257485274659866\\
-14	0.278579731256004\\
};
\addlegendentry{GCB}

\addplot [color=mycolor1,  mark=+, mark size=2.5pt,mark options={solid, mycolor1}] 
  table[row sep=crcr]{%
-8.5	0.136313364470742\\
-9	0.145391425698583\\
-9.5	0.155344777833051\\
-10	0.166184636273434\\
-10.5	0.177767638599756\\
-11	0.190917084969715\\
-11.5	0.205216958363582\\
-12	0.220738155590132\\
-12.5	0.2383736562626\\
-13	0.257520225742886\\
-13.5	0.278749588938412\\
-14	0.302072838898172\\
};
\addlegendentry{SCB}

\addplot [color=mycolor2, mark=diamond,  mark size=2.5pt,mark options={solid, mycolor2}]  
  table[row sep=crcr]{%
-8.5	0.131242777579059\\
-9	0.140147541303204\\
-9.5	0.149906637611548\\
-10	0.160812520242258\\
-10.5	0.17313000895281\\
-11	0.186586887713651\\
-11.5	0.201950488981829\\
-12	0.217695506001693\\
-12.5	0.235854474906314\\
-13	0.254987581396951\\
-13.5	0.277875271779745\\
-14	0.302279120460986\\
};
\addlegendentry{MUSIC}

\addplot [color=mycolor3, dashed,  mark=triangle,  mark size=2.5pt, mark options={solid, rotate=180, mycolor3}] 
  table[row sep=crcr]{%
-8.5	0.127205365029501\\
-9	0.136514025085071\\
-9.5	0.146695543295587\\
-10	0.157856396173823\\
-10.5	0.170116064357772\\
-11	0.183594369698174\\
-11.5	0.198432683177619\\
-12	0.214839647229154\\
-12.5	0.232876007083811\\
-13	0.252957926939259\\
-13.5	0.275493704863958\\
-14	0.30097882476559\\
};
\addlegendentry{R-MUSIC}

\addplot [color=red, dashed, mark=square,  mark options={solid, red}] 
  table[row sep=crcr]{%
-8.5	0.13156493960525\\
-9	0.139678201592088\\
-9.5	0.148375649394816\\
-10	0.158419274921541\\
-10.5	0.169302490629446\\
-11	0.180796386394567\\
-11.5	0.193917508234817\\
-12	0.208513788512894\\
-12.5	0.224068441924932\\
-13	0.241464973305308\\
-13.5	0.261143638635905\\
-14	0.28303238919483\\
};
\addlegendentry{IAA}

\addplot [color=black, line width=0.8pt]
  table[row sep=crcr]{%
-8.5	0.123503866390242\\
-9	0.132307570848616\\
-9.5	0.141889616302339\\
-10	0.152338192050643\\
-10.5	0.163753064562444\\
-11	0.176247107457715\\
-11.5	0.189948014139885\\
-12	0.205000211678824\\
-12.5	0.221566996525787\\
-13	0.239832915043904\\
-13.5	0.260006414762421\\
-14	0.2823227958124\\
};
\addlegendentry{CRLB}

\end{axis}
\end{tikzpicture}%
 }
\caption{a) percentage of misdetection vs SNR; b) RMSE  of $\hat \theta_1$ (deg.) versus SNR; $K=4$, $L=125$, $N=20$,  $M=1801$.} \label{fig:source1}
\end{figure}

We compare the proposed  GCB method against the Cram\'er-Rao lower bound (CRLB) \cite[Ch. 8.4.2]{van_trees:2002}, SCB,  MUSIC  R-MUSIC, and IAA which were chosen as they share similar computational complexity ({\it cf.} \autoref{tab:running_times}). All methods, except R-MUSIC,  use a grid of DOAs that sets the angular resolution limit of the technique.  In our set-up we have $K=4$ sources at DOAs  $\theta_1 = -30.1^o$, $\theta_2 = -20.02^o$, $\theta_3=-10.02^o$ and $\theta_4=3.02^o$, i.e.,  all except the  $1^{\text{st}}$  source is off the predefined grid.  The SNR of the last 3 sources are $-1$, $-2$ and $-5$ dB  relative to the $1^{\text{st}}$ source.  
The DOA estimates $\hat \theta_1, \ldots, \hat \theta_4$ are ordered by finding the nearest neighbor 
first for  the largest magnitude source $\theta_1$, then for $\theta_2$,  followed by $\theta_3$ and the remaining DOA estimate is paired with smallest magnitude source $\theta_4$. 
We then calculated the empirical root mean squared error (RMSE) $\| \hat{\boldsymbol{\theta}}  - \boldsymbol{\theta}  \|$ averaged over all MC trials. We also calculated RMSE for each source $| \hat \theta_i  - \theta_i |$. 

\begin{figure}[!t]
\centerline{
%
%
\definecolor{mycolor1}{rgb}{1.00000,0.00000,1.00000}%
\definecolor{mycolor2}{rgb}{0.34200,0.87200,0.50300}%
\definecolor{mycolor3}{rgb}{0.92900,0.69400,0.12500}%
\begin{tikzpicture}

\begin{axis}[%
width= 0.44\columnwidth, 
height= 4.3cm, 
 y tick label style={font=\scriptsize} , 
 x tick label style={font=\scriptsize} , 
scale only axis,
xmin=25,
xmax=85,
xlabel style={font=\scriptsize\color{white!15!black}},
xlabel={sample size, $L$},
ymode=log,
ymin=0.317064891127055,
ymax=18.7385284765551,
yminorticks=true,
ylabel style={font=\color{white!15!black}},
axis background/.style={fill=white},
title style={font=\bfseries},
xmajorgrids,
ymajorgrids,
yminorgrids,
legend style={legend cell align=left, align=left, draw=white!15!black,,font=\footnotesize,at={(1.0,1.0)},opacity=0.9}
]
\addplot [color=blue, line width=0.6pt, mark size=2.1pt, mark=o, mark options={solid, blue}]
  table[row sep=crcr]{%
25	1.74034609584799\\
35	0.686979766805399\\
45	0.458462357596927\\
55	0.414703106651814\\
65	0.38121385074522\\
75	0.355189714565796\\
85	0.336948462923731\\
};
\addlegendentry{GCB}

\addplot [color=mycolor1, line width=0.6pt, mark size=2.9pt, mark=+, mark options={solid, mycolor1}]
  table[row sep=crcr]{%
25	18.7385284765551\\
35	4.80801168883771\\
45	1.28208725652092\\
55	0.507837506820177\\
65	0.435939598262572\\
75	0.394696001162075\\
85	0.365047028933715\\
};
\addlegendentry{SCB}

\addplot [color=mycolor2, line width=0.6pt, mark size=2.7pt, mark=diamond, mark options={solid, mycolor2}]
  table[row sep=crcr]{%
25	5.85169180095237\\
35	2.63583398566754\\
45	0.894170378991984\\
55	0.616158312989987\\
65	0.393577099604808\\
75	0.363867375106186\\
85	0.340899203088928\\
};
\addlegendentry{MUSIC}

\addplot [color=mycolor3, dashed, line width=0.6pt, mark size=2.7pt, mark=triangle, mark options={solid, rotate=180, mycolor3}]
  table[row sep=crcr]{%
25	8.18952250969196\\
35	3.43519670069888\\
45	1.35468802719532\\
55	0.724342328250883\\
65	0.388785439244359\\
75	0.359014894157538\\
85	0.335329579187877\\
};
\addlegendentry{R-MUSIC}

\addplot [color=red, dashed, line width=0.6pt, mark size=1.9pt, mark=square, mark options={solid, red}]
  table[row sep=crcr]{%
25	3.61499956662054\\
35	1.17378544319934\\
45	0.471397496811342\\
55	0.424519493074228\\
65	0.388200978875632\\
75	0.361896854550205\\
85	0.341782386907224\\
};
\addlegendentry{IAA}

\addplot [color=black, line width=1.0pt]
  table[row sep=crcr]{%
25	0.584638771918522\\
35	0.494109945566194\\
45	0.43576401209727\\
55	0.394163379618146\\
65	0.362577574527373\\
75	0.337541352345851\\
85	0.317064891127055\\
};
\addlegendentry{CRLB}

\end{axis}
\end{tikzpicture}%
 \hspace{-8pt}
%
%
\definecolor{mycolor1}{rgb}{1.00000,0.00000,1.00000}%
\definecolor{mycolor2}{rgb}{0.34200,0.87200,0.50300}%
\definecolor{mycolor3}{rgb}{0.92900,0.69400,0.12500}%
\begin{tikzpicture}

\begin{axis}[%
width= 0.44\columnwidth, 
height= 4.3cm, 
 y tick label style={font=\scriptsize} , 
 x tick label style={font=\scriptsize} , 
scale only axis,
xmin=25,
xmax=150,
xlabel style={font=\scriptsize\color{white!15!black}},
xlabel={sample size, $L$},
ymode=log,
ymin=0.321738957179615,
ymax = 20, 
yminorticks=true,
ylabel style={font=\color{white!15!black}},
axis background/.style={fill=white},
title style={font=\bfseries},
xmajorgrids,
ymajorgrids,
yminorgrids,
legend style={legend cell align=left, align=left, draw=white!15!black}
]
\addplot [color=blue, line width=0.6pt, mark size=2.0pt, mark=o, mark options={solid, blue}]
  table[row sep=crcr]{%
25	7.60998965395703\\
35	3.5021074607537\\
45	1.71457769338886\\
55	0.866275014068857\\
65	0.556567456229102\\
75	0.469866363980227\\
85	0.443007750120318\\
95	0.42378500838672\\
125	0.369009846300428\\
150	0.339506897328069\\
};

\addplot [color=mycolor1, line width=0.6pt, mark size=2.9pt, mark=+, mark options={solid, mycolor1}]
  table[row sep=crcr]{%
25	29.4246326626745\\
35	13.404831104245\\
45	6.01311180671039\\
55	2.81540389050429\\
65	1.62766614103344\\
75	0.849877559808864\\
85	0.492470777474833\\
95	0.461415503279491\\
125	0.389193525125996\\
150	0.350045139946263\\
};

\addplot [color=mycolor2, line width=0.6pt, mark size=2.4pt, mark=diamond, mark options={solid, mycolor2}]
  table[row sep=crcr]{%
25	14.0141427184588\\
35	9.18130323356474\\
45	6.08283848215616\\
55	3.69556889621431\\
65	2.4061593463443\\
75	1.78782195235805\\
85	0.936546136966389\\
95	0.6767095881297\\
125	0.380510709441928\\
150	0.344353113339587\\
};

\addplot [color=mycolor3, dashed, line width=0.6pt, mark size=2.9pt, mark=triangle, mark options={solid, rotate=180, mycolor3}]
  table[row sep=crcr]{%
25	18.1895501221025\\
35	11.6161288505679\\
45	7.93823390311093\\
55	4.75251706436286\\
65	3.28130283815943\\
75	2.321578606418\\
85	0.974185666520996\\
95	0.833178870913711\\
125	0.374972791061185\\
150	0.338776713703351\\
};

\addplot [color=red, dashed, line width=0.6pt, mark size=1.8pt, mark=square, mark options={solid, red}]
  table[row sep=crcr]{%
25	13.06351183003\\
35	7.71450341456487\\
45	4.21570449628529\\
55	2.8259849964216\\
65	2.11816042829622\\
75	0.768308358234725\\
85	0.990896900119619\\
95	0.432532233866255\\
125	0.37569313719222\\
150	0.34378423465889\\
};

\addplot [color=black, line width=1.0pt]
  table[row sep=crcr]{%
25	0.788096275465223\\
35	0.66606292034513\\
45	0.587412281584879\\
55	0.531334400526446\\
65	0.488756561790363\\
75	0.455007596787188\\
85	0.427405214604621\\
95	0.404284858993128\\
125	0.352447368950927\\
150	0.321738957179615\\
};

\end{axis}
\end{tikzpicture}%
}
\vspace{-2pt}
\caption{RMSE  of $\hat{\boldsymbol{\theta}}$ (deg.) versus sample size $L$ when SNR is $-7$ dB  (left panel) and  $-9$ dB (right panel); $K=4$, $N=20$,  $M=1801$.}    \label{fig:4source_L}
\vspace{7pt}
\centerline{
%
%
\definecolor{mycolor1}{rgb}{1.00000,0.00000,1.00000}%
\definecolor{mycolor2}{rgb}{0.34200,0.87200,0.50300}%
\definecolor{mycolor3}{rgb}{0.92900,0.69400,0.12500}%
\begin{tikzpicture}

\begin{axis}[%
width= 0.44\columnwidth, 
height= 4.3cm, 
 y tick label style={font=\scriptsize} , 
 x tick label style={font=\scriptsize} , 
 scale only axis,
xmin=-14,
xmax=-8.5,
xlabel style={font=\scriptsize\color{white!15!black}},
xlabel={SNR (dB)},
ymode=log,
ymin=0.326413432578675,
ymax=29.9267370556831,
yminorticks=true,
ylabel style={font=\scriptsize\color{white!15!black}},
axis background/.style={fill=white},
title style={font=\scriptsize},
xmajorgrids,
ymajorgrids,
yminorgrids,
legend style={legend cell align=left, align=left, draw=white!15!black,font=\scriptsize,at={(1.0,0.82)},opacity=0.8}
]
\addplot [color=blue, line width=0.6pt, mark size=1.9pt, mark=o, mark options={solid, blue}]
  table[row sep=crcr]{%
-8.5	0.335490685414661\\
-9	0.357866082960278\\
-9.5	0.383029850881973\\
-10	0.409821424525365\\
-10.5	0.440324123042712\\
-11	0.474550032486916\\
-11.5	0.764228412627187\\
-12	1.17441525875646\\
-12.5	2.03162926408011\\
-13	3.3570078740847\\
-13.5	4.96374459455762\\
-14	7.59893571495376\\
};

\addplot [color=mycolor1, line width=0.6pt, mark size=2.5pt, mark=+, mark options={solid, mycolor1}]
  table[row sep=crcr]{%
-8.5	19.9167462469819\\
-9	20.6673680601409\\
-9.5	21.3835219924127\\
-10	22.1309484658927\\
-10.5	23.0619630647523\\
-11	23.8402337684288\\
-11.5	24.6927822220718\\
-12	25.7368458958487\\
-12.5	26.802732097058\\
-13	27.9038027922121\\
-13.5	28.9622543505623\\
-14	29.9267370556831\\
};

\addplot [color=mycolor2, line width=0.6pt, mark size=2.2pt, mark=diamond, mark options={solid, mycolor2}]
  table[row sep=crcr]{%
-8.5	16.403864418687\\
-9	16.7175747563255\\
-9.5	16.9599539150317\\
-10	17.2889594250203\\
-10.5	17.6248398763412\\
-11	18.0264948561832\\
-11.5	18.6293378017935\\
-12	19.2625090058382\\
-12.5	19.9870706307853\\
-13	21.0708587738295\\
-13.5	22.0704217087032\\
-14	23.0570165314885\\
};

\addplot [color=mycolor3, dashed, line width=0.6pt, mark size=2.2pt, mark=triangle, mark options={solid, rotate=180, mycolor3}]
  table[row sep=crcr]{%
-8.5	24.3595738322426\\
-9	24.6798305377244\\
-9.5	24.8997237675932\\
-10	25.2390721823223\\
-10.5	25.5369439362745\\
-11	25.9251764856195\\
-11.5	26.2658916116614\\
-12	26.6011715974539\\
-12.5	26.9175783574349\\
-13	27.4927116948958\\
-13.5	28.0186462274819\\
-14	28.8836395714465\\
};

\addplot [color=red, dashed, line width=0.6pt, mark size=1.7pt, mark=square, mark options={solid, red}]
  table[row sep=crcr]{%
-8.5	0.341243998726229\\
-9	0.362908987672299\\
-9.5	0.387364255793776\\
-10	0.414639924110868\\
-10.5	1.00670492201042\\
-11	2.17340881259525\\
-11.5	3.90410051100123\\
-12	6.44038975011087\\
-12.5	10.6030260586306\\
-13	15.5196186851782\\
-13.5	21.5586292514158\\
-14	27.7493670318201\\
};

\addplot [color=black, line width=0.8pt]
  table[row sep=crcr]{%
-8.5	0.326413432578675\\
-9	0.352447368950927\\
-9.5	0.381121974064003\\
-10	0.412760088446054\\
-10.5	0.447726141825284\\
-11	0.486431236697659\\
-11.5	0.529338812425413\\
-12	0.576970961121026\\
-12.5	0.629915476887484\\
-13	0.68883373140509\\
-13.5	0.754469481405001\\
-14	0.827658727316119\\
};

\end{axis}
\end{tikzpicture}%
 \hspace{-8pt}
%
%
\definecolor{mycolor1}{rgb}{1.00000,0.00000,1.00000}%
\definecolor{mycolor2}{rgb}{0.34200,0.87200,0.50300}%
\definecolor{mycolor3}{rgb}{0.92900,0.69400,0.12500}%
\begin{tikzpicture}

\begin{axis}[%
width= 0.44\columnwidth, 
height= 4.3cm, 
 y tick label style={font=\scriptsize} , 
 x tick label style={font=\scriptsize} , 
scale only axis,
xmin=25,
xmax=150,
xlabel style={font=\scriptsize\color{white!15!black}},
xlabel={sample size, $L$},
ymode=log,
ymin=0.321738957179615,
ymax=43.022649224178,
yminorticks=true,
ylabel style={font=\color{white!15!black}},
axis background/.style={fill=white},
title style={font=\bfseries},
xmajorgrids,
ymajorgrids,
yminorgrids,
legend style={legend cell align=left, align=left, draw=white!15!black,font=\footnotesize,at={(0.97,0.77)}}
]
\addplot [color=blue, line width=0.6pt, mark size=1.9pt, mark=o, mark options={solid, blue}]
  table[row sep=crcr]{%
25	4.55681037569044\\
35	2.00175000104076\\
45	0.6848391538651\\
55	0.526631686602063\\
65	0.479943607798527\\
75	0.448683407315225\\
85	0.423840536051001\\
95	0.402549872686602\\
125	0.357866082960278\\
150	0.330127248193782\\
};
\addlegendentry{GCB}

\addplot [color=mycolor1, line width=0.6pt, mark size=2.4pt, mark=+, mark options={solid, mycolor1}]
  table[row sep=crcr]{%
25	43.022649224178\\
35	37.4376803234388\\
45	35.1833878906888\\
55	32.9162311269076\\
65	30.8022010360731\\
75	28.7486357427038\\
85	27.1115673664704\\
95	25.7153372160143\\
125	20.6673680601409\\
150	17.4065978908382\\
};
\addlegendentry{SCB}

\addplot [color=mycolor2, line width=0.6pt, mark size=2.2pt, mark=diamond, mark options={solid, mycolor2}]
  table[row sep=crcr]{%
25	23.9683904980427\\
35	21.8836018394291\\
45	20.165170014987\\
55	19.2729378870996\\
65	18.3753592871904\\
75	18.1040933713898\\
85	17.7272066158208\\
95	17.0429161432739\\
125	16.7175747563255\\
150	16.6677389148418\\
};
\addlegendentry{MUSIC}

\addplot [color=mycolor3, dashed, line width=0.6pt, mark size=2.2pt, mark=triangle, mark options={solid, rotate=180, mycolor3}]
  table[row sep=crcr]{%
25	29.6101667922804\\
35	28.1914805653631\\
45	27.1360929832456\\
55	26.807390977431\\
65	26.3388336451185\\
75	26.1023750001814\\
85	25.5516130876316\\
95	25.0576856752844\\
125	24.6798305377244\\
150	24.6265426608173\\
};
\addlegendentry{R-MUSIC}

\addplot [color=red, dashed, line width=0.6pt, mark size=1.7pt, mark=square, mark options={solid, red}]
  table[row sep=crcr]{%
25	8.47553604991053\\
35	4.8888308077358\\
45	2.04495685366057\\
55	1.32258529655621\\
65	0.766748807193944\\
75	0.458563845064131\\
85	0.432235892385936\\
95	0.411099014837058\\
125	0.362908987672299\\
150	0.334176201027741\\
};
\addlegendentry{IAA}

\addplot [color=black, line width=0.8pt]
  table[row sep=crcr]{%
25	0.788096275465223\\
35	0.66606292034513\\
45	0.587412281584879\\
55	0.531334400526446\\
65	0.488756561790363\\
75	0.455007596787188\\
85	0.427405214604621\\
95	0.404284858993128\\
125	0.352447368950927\\
150	0.321738957179615\\
};
\addlegendentry{CRLB}

\end{axis}
\end{tikzpicture}%
} 
\vspace{-2pt}
\caption{Correlated constant modulus sources scenario. Left: RMSE versus SNR when sample size is $L=125$.  Right: RMSE  versus sample size $L$ when SNR is $-9$dB; $K=4$, $N=20$,  $M=1801$. }    \label{fig:4source_corr}
 \vspace{7pt}
\centerline{
%
%
\definecolor{mycolor1}{rgb}{1.00000,0.00000,1.00000}%
\definecolor{mycolor2}{rgb}{0.34200,0.87200,0.50300}%
\definecolor{mycolor3}{rgb}{0.92900,0.69400,0.12500}%
\begin{tikzpicture}

\begin{axis}[%
width= 0.44\columnwidth, 
height= 4.3cm, 
 y tick label style={font=\scriptsize} , 
 x tick label style={font=\scriptsize} , 
 scale only axis,
xmin=1,
xmax=15,
xlabel style={font=\scriptsize\color{white!15!black}},
xlabel={angular separation, $\delta_{\theta}$},
ymode=log,
ymin=0.104713154609571,
ymax=70.5451163459712,
yminorticks=true,
ylabel style={font=\scriptsize\color{white!15!black}},
axis background/.style={fill=white},
title style={font=\bfseries},
xmajorgrids,
ymajorgrids,
yminorgrids,
legend style={legend cell align=left, align=left, draw=white!15!black,font=\footnotesize,at={(1.0,0.99)}}
]
\addplot [color=blue, line width=0.6pt, mark size=1.9pt, mark=o, mark options={solid, blue}]
  table[row sep=crcr]{%
1	48.2334973312116\\
2	41.0058106646688\\
3	21.4765091607862\\
5	3.69579790212255\\
7	0.246221580424358\\
9	0.106662083234858\\
12	0.151133936184655\\
15	0.110939022290025\\
};
\addlegendentry{GCB}

\addplot [color=mycolor1, line width=0.6pt, mark size=2.4pt, mark=+, mark options={solid, mycolor1}]
  table[row sep=crcr]{%
1	48.7678692296475\\
2	47.9408091045615\\
3	46.4127171983426\\
5	5.48747177972989\\
7	0.139985237316895\\
9	0.121825011115671\\
12	0.127483332243868\\
15	0.121104913195131\\
};
\addlegendentry{SCB}

\addplot [color=mycolor2, line width=0.6pt, mark size=2.3pt, mark=diamond, mark options={solid, mycolor2}]
  table[row sep=crcr]{%
1	46.9158657968355\\
2	18.2744985229874\\
3	10.2115125291669\\
5	0.191685506320466\\
7	0.128722958325236\\
9	0.116354630333305\\
12	0.118760824068097\\
15	0.113789864809364\\
};
\addlegendentry{MUSIC}

\addplot [color=mycolor3, dashed, line width=0.6pt, mark size=2.3pt, mark=triangle, mark options={solid, rotate=180, mycolor3}]
  table[row sep=crcr]{%
1	36.4064563836437\\
2	3.35526638491864\\
3	0.342785013394751\\
5	0.17596571967568\\
7	0.122805063217958\\
9	0.108154560871437\\
12	0.111917277115344\\
15	0.106347427492552\\
};
\addlegendentry{R-MUSIC}

\addplot [color=red, dashed, line width=0.6pt, mark size=1.7pt, mark=square, mark options={solid, red}]
  table[row sep=crcr]{%
1	70.5451163459712\\
2	70.2323284307161\\
3	68.9368128128939\\
5	46.5017364464316\\
7	0.260824589842801\\
9	0.149854151316095\\
12	0.125849645741787\\
15	0.117235375775971\\
};
\addlegendentry{IAA}

\addplot [color=black, line width=0.8pt]
  table[row sep=crcr]{%
1	1.38244776369384\\
2	0.504826727411272\\
3	0.304608562920055\\
5	0.168657203091345\\
7	0.120542007830761\\
9	0.107144191050822\\
12	0.110535371306102\\
15	0.104713154609571\\
};
\addlegendentry{CRLB}

\end{axis}
\end{tikzpicture}%
  \hspace{-8pt} 
%
%
\definecolor{mycolor1}{rgb}{1.00000,0.00000,1.00000}%
\definecolor{mycolor2}{rgb}{0.34200,0.87200,0.50300}%
\definecolor{mycolor3}{rgb}{0.92900,0.69400,0.12500}%
\begin{tikzpicture}

\begin{axis}[%
width= 0.44\columnwidth, 
height= 4.3cm, 
 y tick label style={font=\scriptsize} , 
 x tick label style={font=\scriptsize} , 
 scale only axis,
xmin=1,
xmax=15,
xlabel style={font=\scriptsize\color{white!15!black}},
xlabel={angular separation, $\delta_{\theta}$},
ymode=log,
ymin=0.150115066976414,
ymax=70.7711998400477,
yminorticks=true,
ylabel style={font=\scriptsize\color{white!15!black}},
axis background/.style={fill=white},
title style={font=\bfseries},
xmajorgrids,
ymajorgrids,
yminorgrids,
legend style={legend cell align=left, align=left, draw=white!15!black}
]
\addplot [color=blue, line width=0.6pt, mark size=1.9pt, mark=o, mark options={solid, blue}]
  table[row sep=crcr]{%
1	48.425663362863\\
2	46.5785016661835\\
3	39.6811404170797\\
5	11.0143233170871\\
7	0.290792480416311\\
9	0.150115066976414\\
12	0.201494085934717\\
15	0.157619372751787\\
};

\addplot [color=mycolor1, line width=0.6pt, mark size=2.4pt, mark=+, mark options={solid, mycolor1}]
  table[row sep=crcr]{%
1	48.8334494815456\\
2	48.582041651897\\
3	47.5779112362028\\
5	32.9880810920148\\
7	0.205604798906383\\
9	0.173435867109431\\
12	0.185831106115204\\
15	0.173743872793642\\
};

\addplot [color=mycolor2, line width=0.6pt, mark size=2.3pt, mark=diamond, mark options={solid, mycolor2}]
  table[row sep=crcr]{%
1	48.80015139594\\
2	36.9135778018514\\
3	16.8962584418364\\
5	0.44574491958218\\
7	0.18778533844082\\
9	0.166518067087829\\
12	0.172271491160513\\
15	0.163566092655742\\
};

\addplot [color=mycolor3, dashed, line width=0.6pt, mark size=2.3pt, mark=triangle, mark options={solid, rotate=180, mycolor3}]
  table[row sep=crcr]{%
1	36.0017733938113\\
2	26.4617744712553\\
3	1.36262935501543\\
5	0.270728130966082\\
7	0.184512213407766\\
9	0.160503180784795\\
12	0.16785622516504\\
15	0.158361402629988\\
};

\addplot [color=red, dashed, line width=0.6pt, mark size=1.7pt, mark=square, mark options={solid, red}]
  table[row sep=crcr]{%
1	70.7711998400477\\
2	70.5828060743783\\
3	69.8846692677773\\
5	65.3799685836579\\
7	0.276010144741098\\
9	0.19226856217281\\
12	0.180693478945237\\
15	0.165933721708397\\
};

\addplot [color=black, line width=0.8pt]
  table[row sep=crcr]{%
1	2.47515250993462\\
2	0.812134295200294\\
3	0.467088400422691\\
5	0.251040069349989\\
7	0.178878229978947\\
9	0.159214542891731\\
12	0.163931545126276\\
15	0.155345571718958\\
};

\end{axis}
\end{tikzpicture}%
} 
\vspace{-2pt}
\caption{RMSE of $\hat{\boldsymbol{\theta}}$ (deg.) versus $\delta_\theta=\theta_2-\theta_1$ in two source scenario when SNR is $-6$ dB (left panel) and $-3$ dB (right panel); $K=2$, $N=20$,  $M=1801$,  $\theta_1 = -30.02^o$.}    \label{fig:4source_angsep}
\end{figure}

{\bf RMSE versus SNR.}  \autoref{fig:4source} displays the performance when the sample size is $L=125$ (left panel)  and $L=25$ (right panel)  and SNR varies.  The proposed GCB exhibits the best performance across all SNR levels and each sample length $L$. At low sample size ($L=25$),  the second-best performing methods is IAA but its performance deteriorates and is similar to MUSIC and R-MUSIC when $L=125$. Additionally, note that for $L=25$, SCB and GCB exhibit opposite performances: SCB performs the worst among all methods, while GCB achieves the best results. This highlights GCB's high robustness to low sample sizes.  At very low SNR, the proposed GCB has superior performance. For example, at SNR of  $-11$ dB and $L=125$, GCB has one order of magnitude smaller RMSE  than IAA.
  
{\bf Percentage of misdetection.}  The main challenge in low SNR conditions is accurately estimating the $4^{\text{th}}$ source, which has the lowest power.   We say that misdetection  (MD) occurs when a DOA estimation error exceeds a threshold   $\tau= \min_{i \neq j} \frac 1 2  | \theta_i - \theta_ j |$. We compute MD using formula   $\mathsf{MD}= I( \max_i | \hat \theta_{(i)} - \theta_i | > \tau )$, where $\hat \theta_{(i)}$ are the ordered DOA estimates ($\hat \theta_{(1)} <  \ldots < \hat \theta_{(4)}$) and $I(\cdot)$ denotes the indicator function. \autoref{fig:source1}a shows the misdetection rate versus SNR for $L=125$. The figure highlights the poor performance of MUSIC and R-MUSIC in low SNR conditions. This is further confirmed in \autoref{fig:source1}b which displays the RMSE of the first (highest power)  DOA estimate $\hat \theta_1$. As can be noted, all methods perform more similarly in estimating $\theta_1$. 
  
{\bf RMSE versus sample size $L$}. \autoref{fig:4source_L} shows the RMSE versus sample size $L$ at SNR -$7$ db (left panel) and $-9$ dB (right panel). GCB consistently exhibits the best performance across all $L$ values.   MUSIC and R-MUSIC catch up a bit with sparsity based GCB and IAA methods when $L$ increases. Additionally,  also SCB outperforms MUSIC and R-MUSIC when SNR is  $-9$ dB.
 
 {\bf Coherent sources scenario.} We generate non-Gaussian sources with constant modulus  by generating the $k^{\text{th}}$ source as $s_k = \gamma_k \exp(\jmath \vartheta_k)$ where the phases $\vartheta_k$, $k=1,\ldots,K$ are independently and uniformly distributed in $[0,2\pi)$ while $\gamma_k$ are as earlier. We set $\vartheta_1=\vartheta_4$, i.e., source 1 and source 4 have identical phases, and are thus fully coherent.  We repeated previous simulation studies and display the  RMSE versus SNR when $L=125$ and RMSE versus sample size $L$ when SNR is $-9$ dB in  \autoref{fig:4source_corr}.  When comparing \autoref{fig:4source_corr} with  \autoref{fig:4source_L} and \autoref{fig:4source}, we can observe that the proposed GCB (as well as IAA) is robust to assumption of uncorrelated sources
 while  MUSIC, R-MUSIC and SCB are not able to localize the 4 sources due to coherence.   Again  GCB performs the best for all SNR levels and all sample sizes $L$. 
 
 {\bf RMSE versus angle separation}: We consider two  sources ($K=2$) 
 with varying 
angle separation  $\delta_{\theta}$.    The DOA of the $1^{\text{st}}$ source is $\theta_1 = -30.02^o$ (off-grid)  
and  $\theta_2 = \theta_1 + \delta_{\theta}$ for $2^{\text{nd}}$ source.  The SNR is $-3$ dB and the sample size $L = 125$. \autoref{fig:4source_angsep} illustrates that  GCB, IAA and SCB are not able to identify close-by sources as accurately as MUSIC an R-MUSIC when $\delta_\theta \leq 6^o$. GCB has better performance over IAA when  $\delta_\theta \leq 7^o$ but perform on par otherwise.

\section{Concluding remarks}   \label{sec:conclusion}
We proposed greedy Capon beamformer that greedily selects a  high-power source using Capon SINR maximization principle and spatial power spectrum.  The power of the selected source is estimated and subsequently the INCM  is updated. The method is computationally light and  performed favourably against SCB and some SOTA methods.



\begin{thebibliography}{10}
\providecommand{\url}[1]{#1}
\csname url@samestyle\endcsname
\providecommand{\newblock}{\relax}
\providecommand{\bibinfo}[2]{#2}
\providecommand{\BIBentrySTDinterwordspacing}{\spaceskip=0pt\relax}
\providecommand{\BIBentryALTinterwordstretchfactor}{4}
\providecommand{\BIBentryALTinterwordspacing}{\spaceskip=\fontdimen2\font plus
\BIBentryALTinterwordstretchfactor\fontdimen3\font minus
  \fontdimen4\font\relax}
\providecommand{\BIBforeignlanguage}[2]{{%
\expandafter\ifx\csname l@#1\endcsname\relax
\typeout{** WARNING: IEEEtran.bst: No hyphenation pattern has been}%
\typeout{** loaded for the language `#1'. Using the pattern for}%
\typeout{** the default language instead.}%
\else
\language=\csname l@#1\endcsname
\fi
#2}}
\providecommand{\BIBdecl}{\relax}
\BIBdecl

\bibitem{van_trees:2002}
H.~L. Van~Trees, \emph{Detection, Estimation and Modulation theory, Part IV:
  Optimum array processing}.\hskip 1em plus 0.5em minus 0.4em\relax New York:
  Wiley, 2002, 1456 pages.

\bibitem{capon:1969}
J.~Capon, ``High resolution frequency-wavenumber spectral analysis,''
  \emph{Proceedings of the {IEEE}}, vol.~57, no.~8, pp. 1408--1418, 1969.

\bibitem{schmidt:1986}
R.~O. Schmidt, ``Multiple emitter location and signal parameter estimation,''
  \emph{{IEEE} Transactions on Antennas and Propagation}, vol.~34, no.~3, pp.
  276--280, 1986.

\bibitem{barabell1983improving}
A.~Barabell, ``Improving the resolution performance of eigenstructure-based
  direction-finding algorithms,'' in \emph{IEEE International Conference on
  Acoustics, Speech, and Signal Processing (ICASSP)}, vol.~8.\hskip 1em plus
  0.5em minus 0.4em\relax IEEE, 1983, pp. 336--339.

\bibitem{yang2018sparse}
Z.~Yang, J.~Li, P.~Stoica, and L.~Xie, ``Sparse methods for
  direction-of-arrival estimation,'' in \emph{Academic Press Library in Signal
  Processing, Volume 7}.\hskip 1em plus 0.5em minus 0.4em\relax Elsevier, 2018,
  pp. 509--581.

\bibitem{krim_viberg:1996}
H.~Krim and M.~Viberg, ``Two decades of array signal processing: the parametric
  approach,'' \emph{{IEEE} Signal Processing Magazine}, vol.~13, no.~4, pp.
  67--94, 1996.

\bibitem{elbir2023twenty}
A.~M. Elbir, K.~V. Mishra, S.~A. Vorobyov, and R.~W. Heath, ``Twenty-five years
  of advances in beamforming: From convex and nonconvex optimization to
  learning techniques,'' \emph{IEEE Signal Processing Magazine}, vol.~40,
  no.~4, pp. 118--131, 2023.

\bibitem{baggeroer1999passive}
A.~Baggeroer and H.~Cox, ``Passive sonar limits upon nulling multiple moving
  ships with large aperture arrays,'' in \emph{Conference Record of the
  Thirty-Third Asilomar Conference on Signals, Systems, and Computers (Cat. No.
  CH37020)}, vol.~1.\hskip 1em plus 0.5em minus 0.4em\relax IEEE, 1999, pp.
  103--108.

\bibitem{mestre2005finite}
X.~Mestre and M.~A. Lagunas, ``Finite sample size effect on minimum variance
  beamformers: Optimum diagonal loading factor for large arrays,'' \emph{IEEE
  Transactions on Signal Processing}, vol.~54, no.~1, pp. 69--82, 2005.

\bibitem{abramovich2007diagonally}
Y.~I. Abramovich and N.~K. Spencer, ``Diagonally loaded normalised sample
  matrix inversion ({LNSMI}) for outlier-resistant adaptive filtering,'' in
  \emph{Proc. Int'l Conf. on Acoustics, Speech and Signal Processing
  ({ICASSP'07})}, Apr. 15--20 2007, pp. 1105--1108.

\bibitem{yang2018high}
L.~Yang, M.~R. McKay, and R.~Couillet, ``High-dimensional mvdr beamforming:
  Optimized solutions based on spiked random matrix models,'' \emph{IEEE
  Transactions on Signal Processing}, vol.~66, no.~7, pp. 1933--1947, 2018.

\bibitem{richmond2020capon}
C.~D. Richmond, ``Capon--bartlett cross-spectrum and a perspective on robust
  adaptive filtering,'' \emph{Signal Processing}, vol. 171, p. 107473, 2020.

\bibitem{wipf2007beamforming}
D.~Wipf and S.~Nagarajan, ``Beamforming using the relevance vector machine,''
  in \emph{Proceedings of the 24th international conference on Machine
  learning}, 2007, pp. 1023--1030.

\bibitem{stoica2010spice}
P.~Stoica, P.~Babu, and J.~Li, ``{SPICE}: A sparse covariance-based estimation
  method for array processing,'' \emph{IEEE Transactions on Signal Processing},
  vol.~59, no.~2, pp. 629--638, 2010.

\bibitem{yardibi2010source}
T.~Yardibi, J.~Li, P.~Stoica, M.~Xue, and A.~B. Baggeroer, ``Source
  localization and sensing: A nonparametric iterative adaptive approach based
  on weighted least squares,'' \emph{{IEEE} Transactions on Aerospace and
  Electronic Systems}, vol.~46, no.~1, pp. 425--443, 2010.

\bibitem{abeida2012iterative}
H.~Abeida, Q.~Zhang, J.~Li, and N.~Merabtine, ``Iterative sparse asymptotic
  minimum variance based approaches for array processing,'' \emph{{IEEE}
  Transactions on Signal Processing}, vol.~61, pp. 933--944, 2012.

\bibitem{gerstoft2016multisnapshot}
P.~Gerstoft, C.~F. Mecklenbr{\"a}uker, A.~Xenaki, and S.~Nannuru,
  ``Multisnapshot sparse {B}ayesian learning for {DOA},'' \emph{{IEEE} Signal
  Processing Letters}, vol.~23, no.~10, pp. 1469--1473, 2016.

\bibitem{mecklenbrauker2024robust}
C.~F. Mecklenbr\"auker, P.~Gerstoft, E.~Ollila, and Y.~Park, ``Robust and
  sparse {M}-estimation of {DOA},'' \emph{Signal Processing}, vol. 220, p.
  109461, 2024.

\bibitem{ollila2024sparse}
E.~Ollila, ``Sparse signal recovery and source localization via covariance
  learning,'' \emph{arXiv preprint, arXiv:2401.13975 [stat.ME]}, 2024.

\bibitem{wang2021sparse}
Q.~Wang, H.~Yu, J.~Li, F.~Ji, and F.~Chen, ``Sparse bayesian learning using
  generalized double {Pareto} prior for {DOA} estimation,'' \emph{IEEE Signal
  Processing Letters}, vol.~28, pp. 1744--1748, 2021.

\bibitem{du2010fully}
L.~Du, J.~Li, and P.~Stoica, ``Fully automatic computation of diagonal loading
  levels for robust adaptive beamforming,'' \emph{{IEEE} Transactions on
  Aerospace and Electronic Systems}, vol.~46, no.~1, pp. 449--458, 2010.

\bibitem{sherman1950adjustment}
J.~Sherman and W.~J. Morrison, ``Adjustment of an inverse matrix corresponding
  to a change in one element of a given matrix,'' \emph{The Annals of
  Mathematical Statistics}, vol.~21, no.~1, pp. 124--127, 1950.

\bibitem{GDP-SBL}
\BIBentryALTinterwordspacing
{CGDP-SBL software}. [Online]. Available:
  \url{https://github.com/WQsen/CGDP-SBL}
\BIBentrySTDinterwordspacing

\bibitem{SBL}
\BIBentryALTinterwordspacing
{SBL software, ver. 4.0}. [Online]. Available:
  \url{https://github.com/gerstoft/SBL}
\BIBentrySTDinterwordspacing

\end{thebibliography}
\end{document}